\documentclass[conference]{IEEEtran}

\usepackage[top=0.7
5in,bottom=0.75in,left=0.75in,right=0.75in]{geometry}
\usepackage[numbers,compress]{natbib}
\usepackage{multicol}
\usepackage[bookmarks=true]{hyperref}
\usepackage{booktabs} 
\usepackage{graphicx} 
\usepackage{gensymb}  
\usepackage{algorithm}
\usepackage[noend]{algpseudocode}
\usepackage{amsmath}
\usepackage{amsthm}
\usepackage{amssymb}
\usepackage{amsfonts}
\usepackage{subcaption}
\usepackage{color}
\usepackage{bm}

\usepackage{tikz}
\usetikzlibrary{math,positioning,automata}

\newtheorem{theorem}{Theorem}
\newtheorem{proposition}{Proposition}
\newtheorem{lemma}{Lemma}
\newtheorem{remark}{Remark}
\newtheorem{definition}{Definition}

\IEEEoverridecommandlockouts  

\begin{document}

\title{
~\\  
Robust Approximate Simulation for Hierarchical Control of Linear Systems under Disturbances}

\author{Vince Kurtz, Patrick M. Wensing, and Hai Lin}

\maketitle

\IEEEpeerreviewmaketitle

\begin{abstract}
    Approximate simulation, an extension of simulation relations from formal methods to continuous systems, is a powerful tool for hierarchical control of complex systems. Finding an approximate simulation relation between the full ``concrete'' system and a simplified ``abstract'' system establishes a bound on the output error between the two systems, allowing one to design a controller for the abstract system while formally certifying performance on the concrete system. However, many real-world control systems are subject to external disturbances, which are not accounted for in the standard approximate simulation framework. We present a notion of robust approximate simulation, which considers external disturbances to the concrete system. We derive output error bounds for the case of linear systems subject to two types of additive disturbances: bounded disturbances and a sequence of (unbounded) impulse disturbances. We demonstrate the need for robust approximate simulation and the effectiveness of our proposed approach with a simulated robot motion planning example.  
\end{abstract}

\section{Introduction}

\subsection{Motivation}

Complex and high-dimensional systems are often difficult to control directly. This leads naturally to hierarchical control strategies, where a simpler (abstract) system model is used in the controller design process. One particularly useful framework for hierarchical control is \textit{approximate simulation}~\cite{girard2009hierarchical}. An approximate simulation relation between the abstract system and the full (concrete) system certifies that the outputs of both systems can remain $\epsilon$-close.  

Approximate simulation relations give rise to control architectures like that shown in Figure \ref{fig:hierarchical_control}. The interface, which maps controls from the abstract system to the concrete system, is designed to enforce $\epsilon$-closeness of the outputs. Given such an interface, we can design a controller for the abstract system and guarantee that the concrete system's output will remain $\epsilon$-close. Furthermore, approximate simulation offers elegant connections to other areas of control theory, as an approximate simulation relation can be certified by finding a Lyapunov-like simulation function, which bounds the output error between the two systems. 

Approximate simulation is a powerful framework for hierarchical controller design. Since it builds off of the notions of simulation and bisimulation relations from formal methods, it can be efficiently applied to discrete transition systems and hybrid systems, as well as continuous systems. Recent results suggest that approximate simulation can be used for control of complex, high-dimensional, and highly nonlinear systems such as legged robots \cite{kurtz2019formal}. However, the traditional notion of approximate simulation does not account for disturbances to the concrete system. This means that any guarantees regarding $\epsilon$-closeness of the outputs may not be valid when disturbances enter the concrete system. This property is especially important when it comes to robotic and cyber-physical systems which operate in the real world, and are thus subject to a variety of disturbances.

In this work, we extend the approximate simulation framework to account for disturbances to the concrete system. We present this extension as a general notion of robust approximate simulation for continuous systems, and provide specific results for two special cases: linear systems subject to bounded additive disturbances, and linear systems subject to additive impulse disturbances. 

\begin{figure}
    \centering
    \begin{tikzpicture}
        \node[draw, minimum height=1cm] (abstract) {Abstract System: $\Sigma_2$};
        \node[draw, below=of abstract, minimum width=5cm] (interface) {Interface: $u_{\mathcal{V}}$};
        \node[draw, below=of interface, minimum height=1cm] (concrete) {Concrete System: $\Sigma_1$};
        
        \draw[<-,thick] (abstract.north) -- node[pos=0.7,left] {$\mathbf{u}_2$} ++(0cm,1cm);
        \draw[->,thick] ([yshift=0.3cm]abstract.north) -| ([xshift=-2.0cm]interface.north);
        \draw[->,thick] (abstract) -- node[pos=0.5,left] {$\mathbf{x}_2$} (interface);
        \draw[->,thick] (interface) -- node[pos=0.5,left] {$\mathbf{u}_1$} (concrete);
        \draw[<-,thick] (concrete.west) -- node[pos=0.8,above] {$\mathbf{d}$} ++(-1.3cm,0cm);
        \draw[->,thick] (concrete.south) -- node[pos=0.7,left] {$\mathbf{x}_1$} ++(0cm,-0.5cm) -| ++ (2.8cm,3.5cm) -| ([xshift=2cm]interface.north);
        
        \draw[->,thick,dashed] (abstract.east) -- node[pos=1,right] {$\mathbf{y}_2$} ++(1.7cm,0cm);
        \draw[->,thick,dashed] (concrete.east) -- node[pos=1,right] {$\mathbf{y}_1$} ++(1.7cm,0cm);
    \end{tikzpicture} 
    \caption{Hierarchical control system architecture considered in this work. We extend the notion of approximate simulation~\cite{girard2011approximate} to account for disturbances $\mathbf{d}$ to the concrete system.}
    \label{fig:hierarchical_control}
\end{figure}
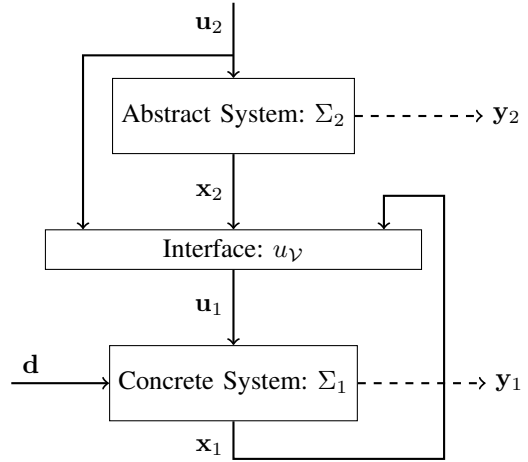

\subsection{Related Work}

The notion of approximate simulation has its roots in the formal methods literature, where exact simulation relations are defined in terms of transition systems \cite{baier2008principles}. Such simulation and bisimulation (both systems simulate each other) relations have had a powerful impact on model checking and formal synthesis, with widespread application in software and design verification \cite{baier2008principles,clarke2018model}. Furthermore, simulation and bisimulation relations have been successfully applied to controller synthesis from temporal logic specifications \cite{da2019active,alonso2018reactive,fainekos2005temporal}. If two systems can be shown to be (bi)similar, the formal design process can consider only the simpler model, significantly improving computational efficiency while also maintaining formal guarantees.

More recently, there has been a growing movement to apply these impactful techniques from formal methods to (continuous) control systems. Large strides have been made in this direction through the use of approximate simulation and approximate bisimulation \cite{girard2005approximate,girard2007approximate,fainekos2007hierarchical}. Since requiring two continuous systems to have exactly the same output may be too strict, approximate (bi)simulation relaxes this requirement to enforce only that the outputs remain $\epsilon$-close. Beyond enabling the application of formal methods techniques to continuous systems, approximate simulation provides an elegant bridge between formal methods and classical control theory: approximate simulation can be equivalently defined in terms of a Lyapunov-like simulation function \cite{girard2011approximate}. 

Hierarchical control and system equivalence have been studied more directly in the context of control systems as well. Notions such as asymptotic model matching \cite{di1994asymptotic} enforce output global asymptotic stability of the joint system \cite{angeli2004uniform}. Since approximate simulation requires only closeness of the system outputs, asymptotic model matching implies approximate simulation but not vice versa \cite{girard2009hierarchical}. Approximate simulation is most closely related to the notion of Input-to-Output Stability (IOS) \cite{sontag1999notions}. Specifically, the simulation function that certifies approximate simulation can be viewed as an IOS Lyapunov function \cite{sontag2000lyapunov} of the joint system. The key difference between IOS and approximate simulation is that the input $\mathbf{u}_2$ to the abstract system is considered a control parameter in the framework of approximate simulation, while it is treated as an unknown disturbance in the IOS framework \cite{girard2009hierarchical}. Furthermore, approximate simulation provides a connection with simulation relations for more general transition systems. Given this emphasis, the key weakness of the standard approximate simulation framework is that disturbances to the concrete system are not considered. Drawing inspiration from the robust control literature, we address this gap by proposing the notion of \textit{robust approximate simulation}. To the best of our knowledge, this is the first work to extend the notion of approximate simulation to account for external disturbances to the concrete system.

The remainder of this paper is organized as follows. Section~\ref{sec:background} introduces necessary background on approximate simulation, with particular attention to the case of linear systems. Section~\ref{sec:main_results} presents our main results, which include a definition of robust approximate simulation for general continuous systems and detailed derivation of error bounds for linear systems subject to bounded and impulse disturbances. We illustrate these results with a robot motion planning example in Section~\ref{sec:example} and conclude with Section~\ref{sec:conclusion}.

\section{Background}\label{sec:background}

\subsection{Approximate Simulation}\label{subsec:approximate_simulation}
The classical notion of approximate simulation is defined in terms of two systems, $\Sigma_1$ and $\Sigma_2$:
\begin{equation}\label{eq:nonlinear_system}
    \Sigma_1 :
    \begin{cases}
        \dot{\mathbf{x}}_1 = f_1(\mathbf{x}_1,\mathbf{u}_1) \\
        \mathbf{y}_1 = g_1(\mathbf{x}_1)
    \end{cases},~~
    \Sigma_2 :
    \begin{cases}
        \dot{\mathbf{x}}_2 = f_2(\mathbf{x}_2,\mathbf{u}_2) \\
        \mathbf{y}_2 = g_2(\mathbf{x}_2)
    \end{cases},
\end{equation}
where $\mathbf{x}_i \in \mathbb{R}^{n_i}$ are the system states, $\mathbf{u}_i \in \mathbb{R}^{p_i}$ are the control inputs, and $\mathbf{y}_i \in \mathbb{R}^m$ are the system outputs. Note that the states may be different sizes but the outputs must be the same size. Without loss of generality, we consider $\Sigma_1$ to be a more complex ``concrete'' model and $\Sigma_2$ to be the simpler ``abstract'' model. This typically means that $n_2 < n_1$.  

Approximate simulation is defined in terms of a Lyapunov-like simulation function $\mathcal{V}$ and an interface function $u_\mathcal{V}$:
\begin{definition}[\citet{girard2011approximate}]\label{definition:simulation_function}
    Consider two systems of the form (\ref{eq:nonlinear_system}). Let $\mathcal{V} : \mathbb{R}^{n_1} \times \mathbb{R}^{n_2} \to \mathbb{R}^+$ be a smooth function and $u_\mathcal{V} : \mathbb{R}^{p_2} \times \mathbb{R}^{n_1} \times \mathbb{R}^{n_2} \to \mathbb{R}^{p_1}$ be a continuous function. $\mathcal{V}$ is a simulation function of $\Sigma_2$ by $\Sigma_1$ and $u_\mathcal{V}$ is an associated interface if there exists a class-$\kappa$ function\footnote{A function $\gamma : \mathbb{R}^+ \to \mathbb{R}^+$ is a class-$\kappa$ function if it is continuous, strictly increasing, and $\gamma(0)=0$.} $\gamma$ such that for all $\mathbf{x}_1, \mathbf{x}_2 \in \mathbb{R}^{n_1} \times \mathbb{R}^{n_2}$,
    \begin{equation}\label{eq:output_error_bound}
        \mathcal{V}(\mathbf{x}_1,\mathbf{x}_2) \geq \|g_1(\mathbf{x}_1)-g_2(\mathbf{x}_2)\|,
    \end{equation}
    and for all $\mathbf{u}_2 \in \mathbb{R}^{p_2}$ satisfying $\gamma(\|\mathbf{u}_2\|) < \mathcal{V}(\mathbf{x}_1, \mathbf{x}_2)$,
    \begin{equation}\label{eq:simulation_fcn_decreasing}
        \frac{\partial\mathcal{V}}{\partial\mathbf{x}_2}f_2(\mathbf{x}_2, \mathbf{u}_2) +\frac{\partial\mathcal{V}}{\partial\mathbf{x}_1}f_1(\mathbf{x}_1, u_{\mathcal{V}}(\mathbf{u}_2, \mathbf{x}_1, \mathbf{x}_2)) < 0.
    \end{equation}
\end{definition}
\begin{definition}[\citet{girard2011approximate}]
    $\Sigma_1$ approximately simulates $\Sigma_2$ if there exists a simulation function $\mathcal{V}$ of $\Sigma_2$ by $\Sigma_1$.
\end{definition}

The conditions (\ref{eq:output_error_bound}-\ref{eq:simulation_fcn_decreasing}) essentially state that the simulation function bounds the output error, and as long as the input to the abstract system is not too large, the simulation function will be decreasing. 

If $\Sigma_1$ approximately simulates $\Sigma_2$, we can bound the output error of the two systems as follows:
\begin{theorem}[\citet{girard2009hierarchical}]\label{theorem:bound}
    Consider two systems of the form (\ref{eq:nonlinear_system}). Let $\mathcal{V}$ be a simulation function of $\Sigma_2$ by $\Sigma_1$ and $u_\mathcal{V}$ be an associated interface. Let $\mathbf{u}_2(t)$ be an admissible input of $\Sigma_2$ with associated state and output trajectories $\mathbf{x}_2(t)$ and $\mathbf{y}_2(t)$. Let $\mathbf{x}_1(t)$ be a state trajectory of $\Sigma_1$ satisfying
    \begin{equation*}
        \dot{\mathbf{x}}_1 = f_1(\mathbf{x}_1, u_\mathcal{V}(\mathbf{u}_2, \mathbf{x}_1, \mathbf{x}_2))
    \end{equation*}
    and $\mathbf{y}_1(t)$ be the associated output trajectory. Then 
    \begin{equation*}
        \|\mathbf{y}_1(t)-\mathbf{y}_2(t)\| \leq \epsilon,
    \end{equation*}
    where
    \begin{equation}\label{eq:error_bound}
        \epsilon = \max\big\{ \mathcal{V}\left(\mathbf{x}_1(0),\mathbf{x}_2(0)\right), \gamma(\|\mathbf{u}_2\|_\infty)\big\}.
    \end{equation}
\end{theorem}


\subsection{Approximate Simulation for Linear Systems}\label{subsec:linear_simulation}

Finding a simulation function and an interface for two arbitrary systems is a difficult and open problem, though some promising results with sum-of-squares programming do exist \cite{girard2005approximate,murthy2015computing}. For linear systems, however, there are well-defined conditions for the existence of a simulation function \cite{girard2009hierarchical}, which we summarize below.

Consider the case when both the concrete and the abstract systems are linear, i.e.,
\begin{equation}\label{eq:linear_systems}
    \Sigma_i :
    \begin{cases}
        \dot{\mathbf{x}}_i = A_i\mathbf{x}_i + B_i\mathbf{u}_i \\
        \mathbf{y}_i = C_i\mathbf{x}_i
    \end{cases}, ~~~ i = \{1,2\}.
\end{equation}
In this case, there are strong results regarding the form of a simulation function. First note the following Lemma:
\begin{lemma}[\citet{girard2007approximate}]\label{lemma:M}
    If $\Sigma_1$ is stabilizable with feedback gain $K$, i.e. $(A_1 + B_1K)$ is Hurwitz, then there exists a positive definite symmetric matrix $M$ and positive scalar constant $\lambda$ such that the following hold:
    \begin{gather}
        M \geq C_1^TC_1,\label{eq:M_bound}, \\
        (A_1+B_1K)^TM + M(A_1+B_1K) \leq -2 \lambda M.\label{eq:M_lambda_decrease}
    \end{gather}
\end{lemma}
Such an $M$ can be used to show exponential convergence of $\mathbf{y}_1$ to zero with rate $\lambda$ under the feedback $\mathbf{u}_1 = K\mathbf{x}_1$. Note that $M$ and $K$ can be computed jointly using semi-definite programming \cite{girard2009hierarchical}: letting $\bar{K} = KM^{-1}$ and $\bar{M} = M^{-1}$, we have the equivalent linear matrix inequality conditions
\begin{gather}
    \begin{bmatrix}
        \bar{M} & \bar{M}C_1^T \\
        C_1\bar{M} & I 
    \end{bmatrix} \geq 0, \\
    \bar{M}A_1^T + A_1\bar{M} + \bar{K}^TB_1^T + B_1\bar{K} \leq -2 \lambda \bar{M}.
\end{gather}

We can now state the following Theorem:

\begin{theorem}[\citet{girard2009hierarchical}]\label{theorem:simulation}
    Consider two systems of the form (\ref{eq:linear_systems}). Assume that $\Sigma_1$ is stabilizable with feedback gain $K$ and that there exist matrices $P \in \mathbb{R}^{n_1 \times n_2}$ and $Q \in \mathbb{R}^{m_1 \times n_2}$ such that the following conditions hold:
    \begin{align}
        PA_2 &= A_1P + B_1Q, \\
        C_2 & = C_1P.
    \end{align}
    Then a simulation function of $\Sigma_2$ by $\Sigma_1$ is given by
    \begin{equation}
        \mathcal{V}(\mathbf{x}_1,\mathbf{x}_2) = \sqrt{ (\mathbf{x}_1-P\mathbf{x}_2)^TM(\mathbf{x}_1-P\mathbf{x}_2}),
    \end{equation}
    an associated interface is
    \begin{equation}
        u_{\mathcal{V}} = R\mathbf{u}_2 + Q\mathbf{x}_2 + K(\mathbf{x}_1 - P\mathbf{x}_2),
    \end{equation}
    and the class-$\kappa$ function $\gamma$ is given by
    \begin{equation}\label{eq:gamma_fcn}
        \gamma(\nu) = \frac{\|\sqrt{M}(B_1R-PB_2)\|}{\lambda}\nu,
    \end{equation}
    where $R$ is an arbitrary matrix of proper dimensions and $M,\lambda$ satisfy (\ref{eq:M_bound}-\ref{eq:M_lambda_decrease})\footnote{Throughout this text, $\|A\|, A \in \mathbb{R}^{k\times l}$ refers to the induced 2-norm $\|A\| = \sup \{\|Ax\|_2 : x \in \mathbb{R}^l \text{ and } \|x\|_2=1\}$.}.
\end{theorem}

The matrix $R$ acts as a ``feedforward'' mapping from $\mathbf{u}_2$ to $\mathbf{u}_1$. While the simulation relation holds for any $R$ of proper dimensions, choosing $R$ to minimize (\ref{eq:gamma_fcn}) is a logical choice, as this tightens the error bound (\ref{eq:error_bound}).

\section{Main Results}\label{sec:main_results}

\subsection{Robust Approximate Simulation}

Consider the systems 
\begin{equation}\label{eq:disturbance_system}
    \Sigma_1 :
    \begin{cases}
        \dot{\mathbf{x}}_1 = f_1(\mathbf{x}_1,\mathbf{u}_1, \mathbf{d}) \\
        \mathbf{y}_1 = g_1(\mathbf{x}_1)
    \end{cases},~~~
    \Sigma_2 :
    \begin{cases}
        \dot{\mathbf{x}}_2 = f_2(\mathbf{x}_2,\mathbf{u}_2) \\
        \mathbf{y}_2 = g_2(\mathbf{x}_2)
    \end{cases},
\end{equation}
where $\mathbf{x}_i \in \mathbb{R}^{n_i}$ are the system states, $\mathbf{u}_i \in \mathbb{R}^{p_i}$ are the control inputs, $\mathbf{y}_i \in \mathbb{R}^m$ are the system outputs, and $\mathbf{d} \in \mathbb{R}^q$ is an unknown disturbance. 

We define robust approximate simulation in terms of a Lyapunov-like robust simulation function as follows:
\begin{definition}[Robust Simulation Function]\label{definition:robust_simulation_function}
    Consider two systems of the form (\ref{eq:disturbance_system}). Let $\mathcal{V} : \mathbb{R}^{n_1} \times \mathbb{R}^{n_2} \to \mathbb{R}^+$ be a smooth function and $u_\mathcal{V} : \mathbb{R}^{p_2} \times \mathbb{R}^{n_1} \times \mathbb{R}^{n_2} \to \mathbb{R}^{p_1}$ be a continuous function. $\mathcal{V}$ is a robust simulation function of $\Sigma_2$ by $\Sigma_1$ and $u_\mathcal{V}$ is an associated interface if there exists class-$\kappa$ functions $\gamma_1$, $\gamma_2$ such that for all $\mathbf{x}_1, \mathbf{x}_2 \in \mathbb{R}^{n_1} \times \mathbb{R}^{n_2}$,
    \begin{equation}\label{eq:error_bounded_robust_as}
        \mathcal{V}(\mathbf{x}_1,\mathbf{x}_2) \geq \|g_1(\mathbf{x}_1)-g_2(\mathbf{x}_2)\|
    \end{equation}
    and for all $\mathbf{u}_2 \in \mathbb{R}^{p_2}$ satisfying $\gamma_1(\|\mathbf{d}\|) + \gamma_2(\|\mathbf{u}_2\|) < \mathcal{V}(\mathbf{x}_1, \mathbf{x}_2)$,
    \begin{equation}\label{eq:Vdot_decreasing_robust_as}
        \frac{\partial\mathcal{V}}{\partial\mathbf{x}_2}f_2(\mathbf{x}_2, \mathbf{u}_2) +\frac{\partial\mathcal{V}}{\partial\mathbf{x}_1}f_1(\mathbf{x}_1, u_{\mathcal{V}}(\mathbf{u}_2, \mathbf{x}_1, \mathbf{x}_2),\mathbf{d}) < 0.
    \end{equation}
\end{definition}
\begin{definition}[Robust Approximate Simulation]
    $\Sigma_1$ robustly approximately simulates $\Sigma_2$ if there exists a robust simulation function $\mathcal{V}$ of $\Sigma_2$ by $\Sigma_1$.
\end{definition}

Note that Definition \ref{definition:robust_simulation_function} is a direct generalization of the typical approximate simulation notion: taking $\mathbf{d} = \bm{0}$, we recover System (\ref{eq:nonlinear_system}) and Definition \ref{definition:simulation_function}. The primary difference between robust approximate simulation and traditional approximate simulation is the conditions under which the simulation function decreases along a trajectory. This suggests that for many cases, a simulation function may also be a robust simulation function, though the resulting error bounds would be different. 

As in the case of conventional approximate simulation, finding a robust simulation function and an interface for two general systems is a difficult problem. In the following subsections, we consider the special cases of linear systems under bounded and impulse disturbances. For each case, we show that the conventional approximate simulation function is also a robust simulation function, and derive the associated error bounds. 

\subsection{Linear Systems under Bounded Disturbances}

Consider the following special case of System (\ref{eq:disturbance_system}):
\begin{align}
    \Sigma_1 &:
    \begin{cases}
        \dot{\mathbf{x}}_1 = A_1\mathbf{x}_1 + B_1\mathbf{u}_1 + B_d \mathbf{d}  \\
        \mathbf{y}_1 = C_1\mathbf{x}_1
    \end{cases}, \nonumber\\
    \Sigma_2 &:
    \begin{cases}
        \dot{\mathbf{x}}_2 = A_2\mathbf{x}_2 + B_2\mathbf{u}_2 \\
        \mathbf{y}_2 = C_2\mathbf{x}_2
    \end{cases}, \label{eq:linear_systems_bounded_disturbance}
\end{align}
where $A_i$, $B_i$, $C_i$, and $B_d$ are matrices of proper dimension, and $\mathbf{d}(t) \in \mathbb{R}^q$ is an external disturbance signal that is unknown but bounded in the sense that $\|\mathbf{d}\|_{\infty} < d_{max}$. We assume that $B_d$, the mapping from disturbances to the system, is known.

Taking inspiration from \cite[Theorem 2]{girard2009hierarchical}, we can establish the following analogue to Theorem \ref{theorem:simulation}:

\begin{theorem}\label{theorem:simulation_bounded}
    Consider two systems of the form (\ref{eq:linear_systems_bounded_disturbance}). Assume that $\Sigma_1$ is stabilizable with feedback gain $K$ and that there exist matrices $P$ and $Q$ such that the following conditions hold:
    \begin{align}\label{eq:PQ_condition_simulation_bounded}
        PA_2 &= A_1P + B_1Q, \\
        C_2 & = C_1P.\label{eq:PQ_output_condition}
    \end{align}
    Then a robust simulation function of $\Sigma_2$ by $\Sigma_1$ is given by
    \begin{equation}
        \mathcal{V}(\mathbf{x}_1,\mathbf{x}_2) = \sqrt{ (\mathbf{x}_1-P\mathbf{x}_2)^TM(\mathbf{x}_1-P\mathbf{x}_2}),
    \end{equation}
    an associated interface is
    \begin{equation}
        u_{\mathcal{V}} = R\mathbf{u}_2 + Q\mathbf{x}_2 + K(\mathbf{x}_1 - P\mathbf{x}_2),
    \end{equation}
    the class-$\kappa$ function $\gamma_1$ is given by
    \begin{equation}
        \gamma_1(\nu) = \frac{\|\sqrt{M}B_d\|}{\lambda}\nu,
    \end{equation}
    and the class-$\kappa$ function $\gamma_2$ is given by
    \begin{equation}
        \gamma_2(\nu) = \frac{\|\sqrt{M}(B_1R-PB_2)\|}{\lambda}\nu,
    \end{equation}
    where $R$ is an arbitrary matrix of proper dimensions and $M,\lambda$ are such that (\ref{eq:M_bound}-\ref{eq:M_lambda_decrease}) hold. 
\end{theorem}

\begin{proof}
    From (\ref{eq:M_bound}) and (\ref{eq:PQ_output_condition}) we have
    \begin{multline*}
        \mathcal{V}(\mathbf{x}_1,\mathbf{x}_2) \geq \sqrt{(\mathbf{x}_1-P\mathbf{x}_2)^TC_1^TC_1(\mathbf{x}_1-P\mathbf{x}_2}) \\
        = \|C_1\mathbf{x}_1-C_2\mathbf{x}_2\|,
    \end{multline*}
    so the output error bound condition (\ref{eq:error_bounded_robust_as}) holds. Next we consider the decay rate condition (\ref{eq:Vdot_decreasing_robust_as}). Using equations (\ref{eq:PQ_condition_simulation_bounded}) and (\ref{eq:M_lambda_decrease}), we can show that
    \begin{multline*}
        \frac{\partial\mathcal{V}}{\partial\mathbf{x}_2}(A_2\mathbf{x}_2+B_2\mathbf{u}_2) +\frac{\partial\mathcal{V}}{\partial\mathbf{x}_1}(A_1\mathbf{x}_1+B_1u_\mathcal{V}+B_d\mathbf{d})\\
        \leq -\lambda\mathcal{V}(\mathbf{x}_1,\mathbf{x}_2) + \|\sqrt{M}B_d\mathbf{d} + \sqrt{M}(B_1R+PB_2)\mathbf{u}_2\| \\
        \leq -\lambda\mathcal{V}(\mathbf{x}_1,\mathbf{x}_2) + \|\sqrt{M}B_d\|\|\mathbf{d}\| + \|\sqrt{M}(B_1R+PB_2)\|\|\mathbf{u}_2\|.
    \end{multline*}
    From this it is clear that as long as
    \begin{equation*}
        \frac{\|\sqrt{M}B_d\|}{\lambda}\|\mathbf{d}\| + \frac{\|\sqrt{M}(B_1R-PB_2)\|}{\lambda}\|\mathbf{u}_2\| < \mathcal{V}(\mathbf{x}_1,\mathbf{x}_2),
    \end{equation*}
    we have
    \begin{equation*}
        \frac{\partial\mathcal{V}}{\partial\mathbf{x}_2}(A_2\mathbf{x}_2+B_2\mathbf{u}_2) +\frac{\partial\mathcal{V}}{\partial\mathbf{x}_1}(A_1\mathbf{x}_1+B_1u_\mathcal{V}+B_d\mathbf{d}) < 0,
    \end{equation*}
    and so the Theorem holds.
\end{proof}

Furthermore, we can establish a modified error bound for this case of bounded disturbances:

\begin{theorem}\label{theorem:bound_bounded}
    Consider two systems of the form (\ref{eq:linear_systems_bounded_disturbance}). Let $\mathcal{V}$ be a robust simulation function of $\Sigma_2$ by $\Sigma_1$ and $u_\mathcal{V}$ be an associated interface. Let $\mathbf{u}_2(t)$ be an admissible input of $\Sigma_2$ with associated state and output trajectories $\mathbf{x}_2(t)$ and $\mathbf{y}_2(t)$. Let $\mathbf{x}_1(t)$ be a state trajectory of $\Sigma_1$ satisfying
    \begin{equation*}
        \dot{\mathbf{x}}_1 = A_1\mathbf{x}_1 + B_1u_\mathcal{V} + B_d\mathbf{d}
    \end{equation*}
    and $\mathbf{y}_1(t)$ be the associated output trajectory. Then 
    \begin{multline}\label{eq:error_bound_bounded}
        \|\mathbf{y}_1(t)-\mathbf{y}_2(t)\| \leq \\
        \max\big\{ \mathcal{V}(\mathbf{x}_1(0),\mathbf{x}_2(0)), \gamma_1(\|\mathbf{d}\|_\infty)+ \gamma_2(\|\mathbf{u}_2\|_\infty)\big\}.
    \end{multline}
\end{theorem}

\begin{proof}
    This proof follows closely from the proof of \cite[Theorem 1]{girard2009hierarchical}. To simplify notation, we will denote $\mathcal{V}(\mathbf{x}_1(t),\mathbf{x}_2(t))$ as $\mathcal{V}(t)$. Let $\epsilon = \max\big\{ \mathcal{V}(0), \gamma_1(\|\mathbf{d}\|_\infty) + \gamma_2(\|\mathbf{u}_2\|_\infty)\big\}$. We will show that $\mathcal{V}(t) \leq \epsilon ~~\forall t$. First, note that $\mathcal{V}(0) \leq \epsilon$ trivially. Now assume that there exists $\tau > 0$ such that $\mathcal{V}(\tau) > \epsilon$. Then there also exists some $0 \leq \tau' < \tau$ such that $\mathcal{V}(\tau') = \epsilon$ and $\forall t \in (\tau',\tau], \mathcal{V}(t) > \epsilon$. Note that $\forall t \in (\tau',\tau]$, we have 
    \begin{equation*}
        \gamma_1(\|d\|) + \gamma_2(\|u_2\|) \leq \gamma_1(\|d\|_\infty) + \gamma_2(\|u_2\|_\infty) \leq \epsilon < \mathcal{V}(t).
    \end{equation*}
    From (\ref{eq:Vdot_decreasing_robust_as}), we then have that $\forall t \in (\tau',\tau]$,
    \begin{equation*}
        \frac{d\mathcal{V}(t)}{dt} \leq 0
    \end{equation*}
    which implies
    \begin{equation*}
        \mathcal{V}(\tau)-\mathcal{V}(\tau') = \int_{\tau'}^\tau\frac{d\mathcal{V}(t)}{dt}dt < 0.
    \end{equation*}
    But this contradicts $\mathcal{V}(\tau) > \epsilon = \mathcal{V}(\tau')$. Therefore we must have $\mathcal{V}(t) \leq \epsilon ~~\forall t$. 
    All that is left is to note that $\mathcal{V}(\mathbf{x}_1(t),\mathbf{x}_2(t)) \leq \epsilon \implies \|\mathbf{y}_1(t)-\mathbf{y}_2(t)\| \leq \epsilon$ by Equation~(\ref{eq:error_bounded_robust_as}).
\end{proof}

\subsection{Linear Systems under Impulse Disturbances}\label{sec:problem_form}

Here we consider the case of linear systems under unbounded disturbances that take the form of impulses. We are inspired to consider this type of disturbance model by recent research applying approximate simulation to legged robot locomotion \cite{kurtz2019formal}. In legged locomotion, disturbances from foot impacts with the ground result in infinite-magnitude disturbances over infinitesimally small time periods, and thus cannot be described by the bounded disturbance model described above. 

Consider the following special case of System (\ref{eq:disturbance_system}):
\begin{align}
    \Sigma_1 &:
    \begin{cases}
        \dot{\mathbf{x}}_1 = A_1\mathbf{x}_1 + B_1\mathbf{u}_1 + B_d d  \\
        \mathbf{y}_1 = C_1\mathbf{x}_1
    \end{cases},\nonumber\\
    \Sigma_2 &:
    \begin{cases}
        \dot{\mathbf{x}}_2 = A_2\mathbf{x}_2 + B_2\mathbf{u}_2 \\
        \mathbf{y}_2 = C_2\mathbf{x}_2
    \end{cases}, \label{eq:linear_systems_impulse_disturbance}
\end{align}
where $d(t)$ is a sequence of unit impulses at times $\mathbb{T}$, i.e.,
\begin{equation}
    d(t) = \sum_{t_{i} \in \mathbb{T}} \delta(t-t_{i}),
\end{equation}
where $\delta(t)$ is the Dirac delta function. $B_d(t) \in \mathbb{R}^n$ is a mapping from the impulse disturbances $d(t)$ to the system, which we assume is unknown and possibly time varying, but bounded in the sense that $\|B_d\| < b_{max}$. We make the further assumption that the impulses $t_{i} \in \mathbb{T}$ are separated by at least a minimum dwell time $t_{dwell} > 0$. 

\begin{remark}
    Even though the disturbances we consider with this system model are unbounded, Definition \ref{definition:robust_simulation_function} is flexible enough to handle impulse disturbances. This is because at the instant of an impulse, $d(t)$ is infinite magnitude, and therefore the condition (\ref{eq:Vdot_decreasing_robust_as}) is not enforced. 
\end{remark}

First, we establish a straightforward analogue to Theorem \ref{theorem:simulation} for the the case of the impulse sequence model.

\begin{theorem}\label{theorem:simulation_impulse}
    Consider two systems of the form (\ref{eq:linear_systems_impulse_disturbance}). Assume that $\Sigma_1$ is stabilizable with feedback gain $K$ and that there exist matrices $P$ and $Q$ such that the following conditions hold:
    \begin{align}
        PA_2 &= A_1P + B_1Q, \\
        C_2 & = C_1P.
    \end{align}
    Then a robust simulation function of $\Sigma_2$ by $\Sigma_1$ is given by
    \begin{equation}
        \mathcal{V}(\mathbf{x}_1,\mathbf{x}_2) = \sqrt{ (\mathbf{x}_1-P\mathbf{x}_2)^TM(\mathbf{x}_1-P\mathbf{x}_2}),
    \end{equation}
    an associated interface is
    \begin{equation}
        u_{\mathcal{V}} = R\mathbf{u}_2 + Q\mathbf{x}_2 + K(\mathbf{x}_1 - P\mathbf{x}_2),
    \end{equation}
    the class-$\kappa$ function $\gamma_1$ is given by 
    \begin{equation}
        \gamma_1(\nu) = \nu
    \end{equation}
    and a class-$\kappa$ function $\gamma_2$ is given by
    \begin{equation}
        \gamma_2(\nu) = \frac{\|\sqrt{M}(B_1R-PB_2)\|}{\lambda}\nu,
    \end{equation}
    where $R$ is an arbitrary matrix of proper dimensions and $M,\lambda$ are such that (\ref{eq:M_bound}-\ref{eq:M_lambda_decrease}) hold.
\end{theorem}

\begin{proof}
     Regardless of the disturbance $d$, the condition (\ref{eq:error_bounded_robust_as}) holds trivially following Theorem \ref{theorem:simulation}. For those times when there is not an impulse, i.e., $t \notin \mathbb{T}$, $d(t) = 0 \implies \gamma_1(\|d\|) = 0$ and the second condition (\ref{eq:Vdot_decreasing_robust_as}) also holds following Theorem \ref{theorem:simulation}. For those times when there is an impulse, i.e., $t \in \mathbb{T}$, we have $d(t) = \infty \implies \gamma_1(\|d\|) + \gamma_2(\|\mathbf{u}_2\|) > \mathcal{V}(\mathbf{x}_1,\mathbf{x}_2)$ and so condition (\ref{eq:Vdot_decreasing_robust_as}) is not enforced. With this in mind, note that \textit{any} $\gamma_1(\cdot)$ in class-$\kappa$ is suitable for enforcing robust approximate simulation in this case.
\end{proof}

We can now establish an error bound analogous to that of Theorem \ref{theorem:bound}. To do so, first consider the case of a single impulse disturbance at time $t_i$. For this case, we establish a relaxed upper bound on the output error as follows:

\begin{proposition}\label{proposition:single_impulse}
    Consider two systems of the form (\ref{eq:linear_systems_impulse_disturbance}). Let $\mathcal{V}$ be a robust simulation function of $\Sigma_2$ by $\Sigma_1$ and $u_\mathcal{V}$ be an associated interface. Let $\mathbf{u}_2(t)$ be a smooth admissible input of $\Sigma_2$ with associated state and output trajectories $\mathbf{x}_2(t)$ and $\mathbf{y}_2(t)$. Let $\mathbf{x}_1(t)$ be a state trajectory of $\Sigma_1$ satisfying
    \begin{equation*}
         \dot{\mathbf{x}}_1 = A_1\mathbf{x}_1 + B_1u_\mathcal{V} + B_dd
    \end{equation*}
    and $\mathbf{y}_1(t)$ be the associated output trajectory. Assume that the disturbance signal $d(t)= \delta(t-t_i)$ is a single impulse at time $t_i$. Then 
    \begin{multline}\label{eq:error_bound_single_impulse}
        \|\mathbf{y}_1(t)-\mathbf{y}_2(t)\| \leq \\
        \max\big\{ \mathcal{V}(\mathbf{x}_1(0),\mathbf{x}_2(0)), \gamma(\|\mathbf{u}_2\|_\infty)\big\} +  b_{max}\sqrt{\lambda_{max}},
    \end{multline}
    where $\lambda_{max}$ is the maximum eigenvalue of $M$ and $b_{max}$ is the upper bound on $\|B_d\|$.
\end{proposition}

\begin{proof}
    For simplicity of notation, we denote $\mathcal{V}(\mathbf{x}_1(t),\mathbf{x}_2(t))$ as $\mathcal{V}(t)$ and $\gamma(\|\mathbf{u}_2\|_\infty)$ as $\gamma$.
    We will show that 
    \begin{equation*}
        \mathcal{V}(t) \leq \\
        \max\big\{ \mathcal{V}(0), \gamma \big\} +  b_{max}\sqrt{\lambda_{max}}.
    \end{equation*}

    For $0 \leq t < t_i$, this holds trivially following Theorem \ref{theorem:bound}.
    
    At $t = t_i^+$, we have $\mathbf{x}_1(t_i^+) = \mathbf{x}_1(t_i^-) + B_d$, where $t_i^-$ is the instant immediately before $t_i$ and $t_i^+$ is the instant immediately after. Thus we have
    \begin{align*}
        \mathcal{V}(t_i^+) &= \mathcal{V}\Big(\mathbf{x}_1(t_i^-)+B_d,\mathbf{x}_2(t_i^-)\Big) \\
        &= \|\mathbf{x}_1(t_i^-)+B_d-P\mathbf{x}_2(t_i^-)\|_M \\
        &\leq \max_{\|B_d\| < b_{max}}\|\mathbf{x}_1(t_i^-)+B_d-P\mathbf{x}_2(t_i^-)\|_M \\
        &\leq \max_{\|B_d\| < b_{max}} \Big(\|\mathbf{x}_1(t_i^-)-P\mathbf{x}_2(t_i^-)\|_M + \|B_d\|_M \Big)\\
        &\leq \mathcal{V}(t_i^-) + b_{max}\sqrt{\lambda_{max}} \\
        &\leq \max\big\{ \mathcal{V}(0), \gamma\big\} + b_{max}\sqrt{\lambda_{max}},
    \end{align*}
    where the second inequality follows from the triangle inequality for the inner product defined by $\langle\mathbf{a},\mathbf{b}\rangle = \mathbf{a}^TM\mathbf{b}$ and the corresponding norm $\|\mathbf{a}\|_M = \sqrt{\mathbf{a}^TM\mathbf{a}}$, while the third inequality follows from the smoothness of $\mathbf{u}_2(t)$ and the fact that $\mathbf{a}^TM\mathbf{a} \leq \lambda_{max}\|\mathbf{a}\|$. 
    
    For $t>t_i$ we can follow the proof of Theorem 1 to show that
    \begin{align*}
        \mathcal{V}(t) &\leq \max\big\{ \mathcal{V}(t_i^+), \gamma\big\} \\
        &\leq \max\left\{\max\big\{ \mathcal{V}(0), \gamma\big\} + b_{max}\sqrt{\lambda_{max}}, \gamma\right\} \\
        &\leq \max\big\{ \mathcal{V}(0), \gamma\big\} + b_{max}\sqrt{\lambda_{max}}.
    \end{align*}
    
    Finally, recalling that $\|\mathbf{y}_1(t) -\mathbf{y}_2(t)\| \leq \mathcal{V}(t)$ completes the proof.
\end{proof}

To extend this result to the case of a sequence of impulse disturbances, we recall from the proof of \cite[Theorem 2]{girard2011approximate} that the decay rate of the simulation function is bounded by
\begin{equation*}
    \frac{d\mathcal{V}}{dt} \leq \lambda\big(-\mathcal{V}(\mathbf{x}_1,\mathbf{x}_2)+\gamma(\|\mathbf{u}_2\|)\big).
\end{equation*}
This minimum decay rate allows us to establish a dwell time $t_{dwell}$ such that as long as the impulse disturbances are separated by at least $t_{dwell}$, the simulation function decays enough between impulses that the resulting error bound is the same as that of a single impulse. This is stated formally as follows:

\begin{proposition}\label{proposition:delayed_impulse}
    Consider two systems of the form (\ref{eq:linear_systems_impulse_disturbance}). Let $\mathcal{V}$ be a robust simulation function of $\Sigma_2$ by $\Sigma_1$ and $u_\mathcal{V}$ be an associated interface. Let $\mathbf{u}_2(t)$ be a smooth admissible input of $\Sigma_2$ with associated state and output trajectories $\mathbf{x}_2(t)$ and $\mathbf{y}_2(t)$. Let $\mathbf{x}_1(t)$ be a state trajectory of $\Sigma_1$ satisfying
    \begin{equation*}
        \dot{\mathbf{x}}_1 = A_1\mathbf{x}_1 + B_1u_\mathcal{V} + B_dd
    \end{equation*}
    and $\mathbf{y}_1(t)$ be the associated output trajectory. Assume that $t_{dwell} \geq 1/\lambda$, where $\lambda$ is defined as per Lemma \ref{lemma:M}. Then 
    \begin{multline}\label{eq:error_bound_delayed_impulse}
        \|\mathbf{y}_1(t)-\mathbf{y}_2(t)\| \leq \\
        \max\big\{ \mathcal{V}(\mathbf{x}_1(0),\mathbf{x}_2(0)), \gamma(\|\mathbf{u}_2\|_\infty)\big\} +  b_{max}\sqrt{\lambda_{max}},
    \end{multline}
    where $\lambda_{max}$ is the maximum eigenvalue of $M$ and $b_{max}$ is the upper bound on $\|B_d\|$.
\end{proposition}

\begin{proof}
    We will show that if there is an impact at time $t_i$, the simulation function at time $t_i+t_{dwell}$ is bounded by
    \begin{equation*}
        \mathcal{V}(t_i+t_{dwell}) \leq \max\big\{\mathcal{V}(0),\gamma\big\},
    \end{equation*}
    where $\mathcal{V}(t)$ and $\gamma$ are again shorthand for $\mathcal{V}\big(\mathbf{x}_1(t), \mathbf{x}_2(t)\big)$ and $\gamma(\|\mathbf{u}_2\|_\infty)$ respectively.
    
    For $t_i \leq t \leq t_i+t_{dwell}$, we have the following:
    \begin{align*}
        \frac{d\mathcal{V}}{dt} &\leq \lambda\big(-\mathcal{V}(t)+\gamma\big) \\
        &\leq \lambda\big(-\max\big\{ \mathcal{V}(0),\gamma\big\} - b_{max}\sqrt{\lambda_{max}}+\gamma\big) \\
        &\leq -\lambda b_{max}\sqrt{\lambda_{max}}
    \end{align*}
    since $\mathcal{V}(t) \leq \max\big\{\mathcal{V}(0),\gamma\big\} + b_{max}\sqrt{\lambda_{max}}$ by Proposition \ref{proposition:single_impulse}. We can then compute the simulation function at $t_i+t_{dwell}$ as follows:
    \begin{align*}
        &\mathcal{V}(t_i+t_{dwell}) = \mathcal{V}(t_i) + \int_{t_i}^{t_{i}+t_{dwell}}\frac{d\mathcal{V}}{dt}dt \\
        &\leq \max\big\{ \mathcal{V}(0),\gamma\big\} + b_{max}\sqrt{\lambda_{max}} - t_{dwell}\lambda b_{max} \sqrt{\lambda_{max}} \\
        &\leq \max\big\{ \mathcal{V}(0),\gamma\big\}.
    \end{align*}
    
    Therefore $\mathcal{V}(t_i+t_{dwell}) \leq \max\big\{\mathcal{V}(0),\gamma\big\}$. All that remains is to note that since subsequent impulses occur at $t \geq t_i+t_{dwell}$, the error bound (\ref{eq:error_bound_delayed_impulse}) is enforced for all time by Proposition \ref{proposition:single_impulse}. 
\end{proof}

Interestingly, the minimum dwell time of $1/\lambda$ does not depend on the magnitude of the disturbances $B_d$ but only on $\lambda$, which is essentially the decay rate of the simulation function. A larger disturbance magnitude does increase the error bound, however, through the parameter $b_{max} \geq \|B_d\|$.

Finally, we consider the case of arbitrary dwell times, which may be shorter than $1/\lambda$, and derive an associated error bound:

\begin{theorem}\label{theorem:bound_impulse}
    Consider two systems of the form (\ref{eq:linear_systems_impulse_disturbance}). Let $\mathcal{V}$ be a robust simulation function of $\Sigma_2$ by $\Sigma_1$ and $u_\mathcal{V}$ be an associated interface. Let $\mathbf{u}_2(t)$ be a smooth admissible input of $\Sigma_2$ with associated state and output trajectories $\mathbf{x}_2(t)$ and $\mathbf{y}_2(t)$. Let $\mathbf{x}_1(t)$ be a state trajectory of $\Sigma_1$ satisfying
    \begin{equation*}
        \dot{\mathbf{x}}_1 = A_1\mathbf{x}_1 + B_1u_\mathcal{V} + B_dd
    \end{equation*}
    and $\mathbf{y}_1(t)$ be the associated output trajectory. Assume that $t_{dwell} > 0$. Then 
    \begin{multline}\label{eq:error_bound_impulse}
        \|\mathbf{y}_1(t)-\mathbf{y}_2(t)\| \leq \max\big\{ \mathcal{V}(0), \gamma(\|\mathbf{u}_2\|_\infty)\big\} \\
        +  \max\left\{1,\frac{1}{t_{dwell}\lambda}\right\}b_{max}\sqrt{\lambda_{max}},
    \end{multline}
    where $\lambda_{max}$ is the maximum eigenvalue of $M$, $b_{max}$ is the upper bound on $\|B_d\|$, and $\lambda$ is defined by Lemma \ref{lemma:M}. 
\end{theorem}

\begin{proof}
    Consider the case $0 < t_{dwell} \leq 1/\lambda$. Again, we will establish a bound on $\mathcal{V}(t)$. After the first impact, we have
    \begin{equation*}
        \mathcal{V}(t_1^+) \leq \max\{\mathcal{V}(0),\gamma(\|\mathbf{u}_2\|_\infty)\}
        + b_{max}\sqrt{\lambda_{max}}.
    \end{equation*}
    After the second impact, we have
    \begin{multline*}
        \mathcal{V}(t_2^+) \leq \max\{\mathcal{V}(0),\gamma(\|\mathbf{u}_2\|_\infty)\} \\
        + (2-t_{dwell}\lambda) b_{max}\sqrt{\lambda_{max}}.
    \end{multline*}
    After the $n^{th}$ impact, we have
    \begin{equation*}
        \mathcal{V}(t_n^+) \leq \max\{\mathcal{V}(0),\gamma(\|\mathbf{u}_2\|_\infty)\}
        + X_{n}b_{max}\sqrt{\lambda_{max}},
    \end{equation*}
    where $X_n$ is the $n^{th}$ element of the series defined by $X_1 = 1$, $X_{i+1} = X_i + 1 - (t_{dwell}\lambda)X_i$. Rewriting this series as $X_{i+1} = (1-t_{dwell}\lambda)X_i + 1$, we can see that for $0 < t_{dwell}\lambda \leq 1$, 
    \begin{equation*}
        \limsup_{i \to \infty} X_i = \frac{1}{t_{dwell}\lambda}.
    \end{equation*}
    Therefore if $0 < t_{dwell} \leq 1/\lambda$, then
    \begin{equation*}
        \mathcal{V}(t) \leq \max\{\mathcal{V}(0),\gamma(\|\mathbf{u}_2\|_\infty)\}
        + \frac{1}{t_{dwell}\lambda} b_{max}\sqrt{\lambda_{max}}
    \end{equation*}
    for all $t$. 
    
    For the case of $t_{dwell} > 1/\lambda$, we recover the error bound of Proposition \ref{proposition:delayed_impulse}. Putting these cases together, and noting that $\mathcal{V}(t)$ bounds the output error, the Theorem holds. 
\end{proof}

\section{Example}\label{sec:example}

\begin{figure}
    \centering
    \includegraphics[width=0.8\linewidth]{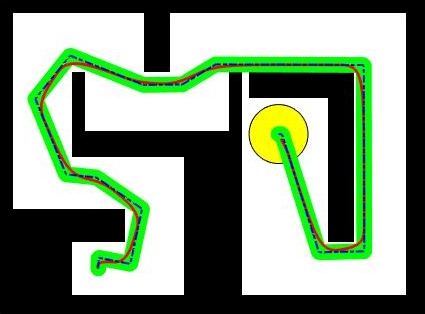}
    \caption{A robot navigates a passageway to reach a goal (yellow circle) without any disturbances. A classical approximate simulation relation guarantees that the concrete system (red solid line) will stay close (green shaded region) to the abstract system trajectory (blue dashed lines). }
    \label{fig:maze_no_disturbance}
\end{figure}

To illustrate the importance of accounting for disturbances when using the framework of approximate simulation, we present a variation of the example presented in \cite{girard2009hierarchical}. In this example, shown in Figure \ref{fig:maze_no_disturbance}, a robot must navigate a narrow passageway before reaching a goal region (yellow). The concrete system is a robot with triple integrator dynamics, i.e.,
\begin{equation*}
    A_1 = 
    \begin{bmatrix}
        0_2 & I_2 & 0_2 \\
        0_2 & 0_2 & I_2 \\
        0_2 & 0_2 & 0_2
    \end{bmatrix}, ~~ 
    B_1 =
    \begin{bmatrix}
        0_2 \\
        0_2 \\
        I_2
    \end{bmatrix}, ~~
    C_1 =
    \begin{bmatrix}
        I_2 & 0_2 & 0_2
    \end{bmatrix},
\end{equation*}
while the abstract system is a single integrator, i.e.,
\begin{equation*}
    A_2 = 0_2, ~~ B_1 = I_2, ~~ C_1 = I_2.
\end{equation*}
The output of both systems represents the position of the robot in the plane. 

Following \cite[Section 5]{girard2009hierarchical}, we found an approximate simulation by choosing
\begin{align*}
    K &= -
    \begin{bmatrix}
        52I_2 & 52.3I_2 & 13I_2
    \end{bmatrix}, \\
    P^T &= \begin{bmatrix}
        I_2 & 0_2 & 0_2
    \end{bmatrix}, \\
    Q &= 0_2, \\
    \lambda &= 1.1,
\end{align*}
and finding $M$ with semi-definite programming (see Section~\ref{subsec:approximate_simulation}). After specifying an a-priori bound on $\|\textbf{u}_2\|$, we computed the associated output error bound (\ref{eq:error_bound}) to be $\epsilon=0.2258$. The abstract system is fully actuated, making it easy to find a trajectory that reaches the goal and stays $\epsilon$ away from all obstacles. In this example, we used the probabilistic roadmap strategy \cite{kavraki1996probabilistic} to find such a trajectory. This is shown in Figure~\ref{fig:maze_no_disturbance} by the blue dashed lines. The green shaded region represents the area in which the concrete system is guaranteed to remain. The associated concrete system trajectory (red solid line) tracks the abstract trajectory effectively, staying within the safe region and eventually arriving at the goal. Note that the output error bound $\epsilon$ is fairly tight. 

\subsection{Bounded Disturbances}

We now consider planning with this same approximate simulation relation under bounded disturbances. Specifically, we choose
\begin{align*}
    B_d^T &= [-0.2~-0.2~0~0~0~0], \\
    d(t) &= 1.
\end{align*}
This represents a constant disturbance pushing the robot down and to the left. This is analogous to what would happen if there was a steady gust of wind pushing the robot when it was deployed in the real world. 

\begin{figure}[H]
    \centering
    \includegraphics[width=0.8\linewidth]{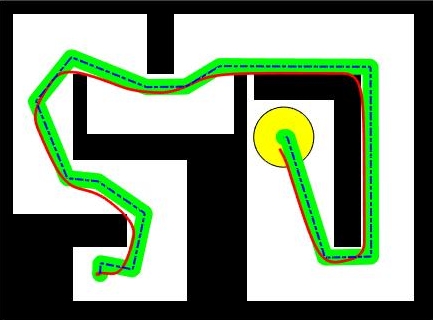}
    \caption{When a bounded disturbance is applied to the concrete system, the classical approximate simulation relation breaks down. The robot leaves the safe region around the planned trajectory, occasionally colliding with obstacles.}
    \label{fig:maze_continuous_fail}
\end{figure}

First, we naively (and improperly) apply the classical approximate simulation relation to this case. The resulting trajectories are shown in Figure \ref{fig:maze_continuous_fail}. The robot is unable to stay $\epsilon$-close to the planned abstract system trajectory, leading to several collisions with the walls of the passageway. 

\begin{figure}[H]
    \centering
    \includegraphics[width=0.8\linewidth]{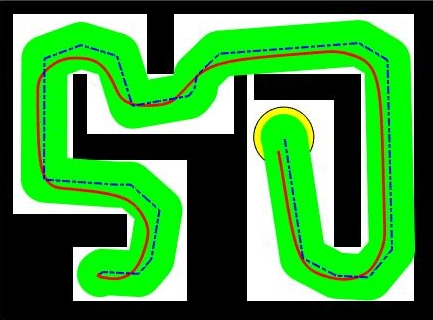}
    \caption{Using the robust simulation relation and the correct output error bound for bounded disturbances (\ref{eq:error_bound_bounded}) results in a more conservative plan for the abstract system and allows the robot to reach the goal. }
    \label{fig:maze_continuous_corrected}
\end{figure}

This motivates the use of a robust approximate simulation relation. Following Theorems \ref{theorem:simulation_bounded} and \ref{theorem:bound_bounded}, we find the correct error bound (\ref{eq:error_bound_bounded}) of $\epsilon=0.6767$. Using this robust approximate simulation relation to plan an abstract system trajectory, we obtain the more conservative results shown in Figure \ref{fig:maze_continuous_corrected}. The robot is able to stay within these relaxed error bounds despite the disturbances, and successfully reaches the goal. We can also see that the associated error bound is reasonably tight. 

\subsection{Impulse Disturbances}

Finally, we consider the case of unbounded (impulse) disturbances. We use the same disturbance mapping $B_d$ as above, and consider impulses occurring every 2.5s. This corresponds to the case of the robot experiencing a push at regular intervals. As in the case of bounded disturbances, improperly applying a classical approximate simulation relation results in the robot leaving the safe region and colliding with obstacles (Figure \ref{fig:maze_impulse_fail}).  

\begin{figure}[H]
    \centering
    \includegraphics[width=0.8\linewidth]{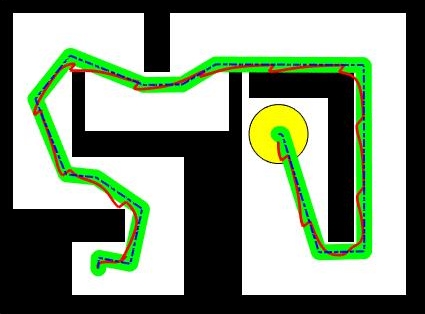}
    \caption{When the concrete system experiences unbounded impulse disturbances, the classical approximate simulation relation breaks down. The robot leaves the safe region around the planned trajectory, occasionally colliding with obstacles.}
    \label{fig:maze_impulse_fail}
\end{figure}

Following Theorem \ref{theorem:bound_impulse}, we find a relaxed error bound of $\epsilon=0.678$. Replanning with this revised margin leads to the safe plan shown in Figure \ref{fig:maze_impulse_corrected}. The robot stays within the safe region, avoids collisions, and reaches the goal.

\begin{figure}[H]
    \centering
    \includegraphics[width=0.8\linewidth]{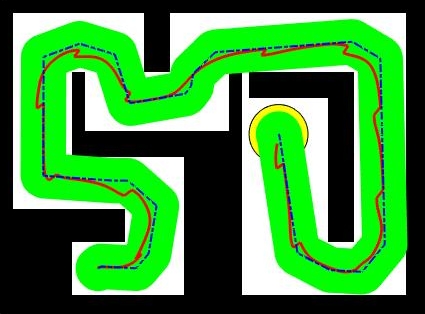}
    \caption{Using the robust simulation relation and the correct output error bound for impulse disturbances (\ref{eq:error_bound_impulse}) results in a more conservative plan for the abstract system and allows the robot to reach the goal. }
    \label{fig:maze_impulse_corrected}
\end{figure}

\section{Conclusion}\label{sec:conclusion}

We proposed a notion of robust approximate simulation as a generalization of approximate simulation. This framework can be used for formally correct hierarchical control in the case when the concrete system is subject to external disturbances. We provided detailed results, including the associated error bounds, for linear systems with two types of disturbances. First, we considered the standard case of arbitrary but bounded disturbances. Second, inspired by impulse disturbances arising during footfalls with legged locomotion, we considered the case where the disturbance signal is a sequence of (unbounded) impulses. This work is one step toward bringing together the best of formal methods and classical control theory to control physical systems. 

Future work will focus on increasing the generality of these results. We anticipate extensions from continuous systems to general transition systems, which will enable us to consider robust approximate  simulation relations between hybrid and discrete-event systems. The noise models we considered in this paper might similarly be extended to the case of finite-energy and stochastic disturbances. Finally, since robust approximate simulation is closely related to Lyapunov stability through the simulation function, we might harness Lyapunov-based techniques like SOS programming to extend these results to nonlinear systems.

\bibliographystyle{IEEEtranN}
{\footnotesize
\bibliography{references}}

\begin{thebibliography}{17}
\providecommand{\natexlab}[1]{#1}
\providecommand{\url}[1]{#1}
\csname url@samestyle\endcsname
\providecommand{\newblock}{\relax}
\providecommand{\bibinfo}[2]{#2}
\providecommand{\BIBentrySTDinterwordspacing}{\spaceskip=0pt\relax}
\providecommand{\BIBentryALTinterwordstretchfactor}{4}
\providecommand{\BIBentryALTinterwordspacing}{\spaceskip=\fontdimen2\font plus
\BIBentryALTinterwordstretchfactor\fontdimen3\font minus
  \fontdimen4\font\relax}
\providecommand{\BIBforeignlanguage}[2]{{%
\expandafter\ifx\csname l@#1\endcsname\relax
\typeout{** WARNING: IEEEtranN.bst: No hyphenation pattern has been}%
\typeout{** loaded for the language `#1'. Using the pattern for}%
\typeout{** the default language instead.}%
\else
\language=\csname l@#1\endcsname
\fi
#2}}
\providecommand{\BIBdecl}{\relax}
\BIBdecl

\bibitem[Girard and Pappas(2009)]{girard2009hierarchical}
A.~Girard and G.~J. Pappas, ``Hierarchical control system design using
  approximate simulation,'' \emph{Automatica}, vol.~45, no.~2, pp. 566--571,
  2009.

\bibitem[Kurtz et~al.(2019)Kurtz, da~Silva, Wensing, and Lin]{kurtz2019formal}
V.~Kurtz, R.~R. da~Silva, P.~M. Wensing, and H.~Lin, ``Formal connections
  between template and anchor models via approximate simulation,'' in
  \emph{IEEE-RAS Conference on Humanoid Robots}, 2019.

\bibitem[Girard and Pappas(2011)]{girard2011approximate}
A.~Girard and G.~J. Pappas, ``Approximate bisimulation: A bridge between
  computer science and control theory,'' \emph{European Journal of Control},
  vol.~17, no. 5-6, pp. 568--578, 2011.

\bibitem[Baier and Katoen(2008)]{baier2008principles}
C.~Baier and J.-P. Katoen, \emph{Principles of model checking}.\hskip 1em plus
  0.5em minus 0.4em\relax MIT Press, 2008.

\bibitem[Clarke~Jr et~al.(2018)Clarke~Jr, Grumberg, Kroening, Peled, and
  Veith]{clarke2018model}
E.~M. Clarke~Jr, O.~Grumberg, D.~Kroening, D.~Peled, and H.~Veith, \emph{Model
  checking}, 2018.

\bibitem[da~Silva et~al.(2019)da~Silva, Kurtz, and Lin]{da2019active}
R.~R. da~Silva, V.~Kurtz, and H.~Lin, ``Active perception and control from
  temporal logic specifications,'' \emph{IEEE Control Systems Letters}, vol.~3,
  no.~4, pp. 1068--1073, 2019.

\bibitem[Alonso-Mora et~al.(2018)Alonso-Mora, DeCastro, Raman, Rus, and
  Kress-Gazit]{alonso2018reactive}
J.~Alonso-Mora, J.~A. DeCastro, V.~Raman, D.~Rus, and H.~Kress-Gazit,
  ``Reactive mission and motion planning with deadlock resolution avoiding
  dynamic obstacles,'' \emph{Autonomous Robots}, vol.~42, no.~4, pp. 801--824,
  2018.

\bibitem[Fainekos et~al.(2005)Fainekos, Kress-Gazit, and
  Pappas]{fainekos2005temporal}
G.~E. Fainekos, H.~Kress-Gazit, and G.~J. Pappas, ``Temporal logic motion
  planning for mobile robots,'' in \emph{International Conference on Robotics
  and Automation}.\hskip 1em plus 0.5em minus 0.4em\relax IEEE, 2005, pp.
  2020--2025.

\bibitem[Girard and Pappas(2005)]{girard2005approximate}
A.~Girard and G.~J. Pappas, ``Approximate bisimulations for nonlinear dynamical
  systems,'' in \emph{Proceedings of the 44th IEEE Conference on Decision and
  Control}.\hskip 1em plus 0.5em minus 0.4em\relax IEEE, 2005, pp. 684--689.

\bibitem[Girard and Pappas(2007)]{girard2007approximate}
------, ``Approximate bisimulation relations for constrained linear systems,''
  \emph{Automatica}, vol.~43, no.~8, pp. 1307--1317, 2007.

\bibitem[Fainekos et~al.(2007)Fainekos, Girard, and
  Pappas]{fainekos2007hierarchical}
G.~E. Fainekos, A.~Girard, and G.~J. Pappas, ``Hierarchical synthesis of hybrid
  controllers from temporal logic specifications,'' in \emph{International
  Workshop on Hybrid Systems: Computation and Control}.\hskip 1em plus 0.5em
  minus 0.4em\relax Springer, 2007, pp. 203--216.

\bibitem[Di~Benedetto and Grizzle(1994)]{di1994asymptotic}
M.~D. Di~Benedetto and J.~Grizzle, ``Asymptotic model matching for nonlinear
  systems,'' \emph{IEEE Transactions on Automatic Control}, vol.~39, no.~8, pp.
  1539--1550, 1994.

\bibitem[Angeli et~al.(2004)Angeli, Ingalls, Sontag, and
  Wang]{angeli2004uniform}
D.~Angeli, B.~Ingalls, E.~Sontag, and Y.~Wang, ``Uniform global asymptotic
  stability of differential inclusions,'' \emph{Journal of Dynamical and
  Control Systems}, vol.~10, no.~3, pp. 391--412, 2004.

\bibitem[Sontag and Wang(1999)]{sontag1999notions}
E.~D. Sontag and Y.~Wang, ``Notions of input to output stability,''
  \emph{Systems \& Control Letters}, vol.~38, no. 4-5, pp. 235--248, 1999.

\bibitem[Sontag and Wang(2000)]{sontag2000lyapunov}
E.~Sontag and Y.~Wang, ``Lyapunov characterizations of input to output
  stability,'' \emph{SIAM Journal on Control and Optimization}, vol.~39, no.~1,
  pp. 226--249, 2000.

\bibitem[Murthy et~al.(2015)Murthy, Islam, Smolka, and
  Grosu]{murthy2015computing}
A.~Murthy, M.~A. Islam, S.~A. Smolka, and R.~Grosu, ``Computing bisimulation
  functions using sos optimization and $\delta$-decidability over the reals,''
  in \emph{Proceedings of the 18th International Conference on Hybrid Systems:
  Computation and Control}.\hskip 1em plus 0.5em minus 0.4em\relax ACM, 2015,
  pp. 78--87.

\bibitem[Kavraki et~al.(1996)Kavraki, Svestka, Latombe, and
  Overmars]{kavraki1996probabilistic}
L.~E. Kavraki, P.~Svestka, J.-C. Latombe, and M.~H. Overmars, ``Probabilistic
  roadmaps for path planning in high-dimensional configuration spaces,''
  \emph{IEEE transactions on Robotics and Automation}, vol.~12, no.~4, pp.
  566--580, 1996.

\end{thebibliography}

\end{document}